\documentclass[letterpaper]{article}
\usepackage[utf8]{inputenc}
\usepackage[english]{babel}
\usepackage{mathtools,amsmath,amssymb,amsbsy}
\usepackage{algorithmic}
\usepackage{amsthm}
\usepackage{aaai_child_ref}
\usepackage{csquotes}
\usepackage{color}
\usepackage{enumitem}
\usepackage{fancyhdr}
\usepackage{dirtytalk}
\usepackage{geometry}
\usepackage{times}
\usepackage{helvet}
\usepackage{courier}
\newtheorem{assumption}{Assumption}
\newtheorem{theorem}{Theorem}
\newtheorem{proposition}{Proposition}
\newtheorem{definition}{Definition}
\newtheorem{corollary}{Corollary}
\newtheorem{lemma}{Lemma}

\begin{document}
\title{Decentralised Learning in Systems with Many, Many Strategic Agents }
\author{David Mguni$^1$, \hspace{1.5 mm} Joel Jennings$^1$ \and Enrique Munoz de Cote$^{1,2}$\\
$^1$PROWLER.io, Cambridge, UK\\
$^2$Department of Computer Science, INAOE, Mexico
}
\maketitle
\begin{abstract}
Although multi-agent reinforcement learning can tackle systems of strategically interacting entities, it currently fails in scalability and lacks rigorous convergence guarantees. Crucially, learning in multi-agent systems can become intractable due to the explosion in the size of the state-action space as the number of agents increases. In this paper, we propose a method for computing closed-loop optimal policies in multi-agent systems that scales \emph{independently} of the number of agents. This allows us to show, for the first time, successful convergence to optimal behaviour in systems with an \emph{unbounded} number of interacting adaptive learners. Studying the asymptotic regime of $N-$player stochastic games, we devise a learning protocol that is guaranteed to converge to equilibrium policies even when the number of agents is extremely large. Our method is \emph{model-free} and completely decentralised so that each agent need only observe its local state information and its realised rewards. We validate these theoretical results by showing convergence to Nash-equilibrium policies in applications from economics and control theory with thousands of strategically interacting agents.
\end{abstract}
\section{Introduction}
Multi-agent reinforcement learning (MARL) provides the potential to systematically analyse environments with strategically interacting agents. Despite the fundamental relevance of multi-agent systems (MASs) with appreciably large populations, learning stable, best-response policies in MASs with more than a few agents remains a significant challenge due to growth in complexity as the number of agents increases \cite{shoham}. 
Consequently, the task of understanding agent behaviour in many systems of interest has been left unaddressed. \newline
\indent This paper seeks to address the problem of learning stable, best-response policies within non-cooperative\footnote{In game theory, the term non-cooperative implies that each agent seeks to pursue its own objectives and agreements between agents over their actions cannot be enforced.} MASs  when the size of the population is large, therefore expanding the range of applications of multi-agent technology.\newline
\indent In a non-cooperative MAS, selfish agents compete to obtain a sequence of rewards within an unknown environment. A stochastic (dynamic) game (SG) is a mathematical framework that analyses the behaviour of strategically interacting entities in non-cooperative settings. Stochastic games enable stable policy outcomes in which agents respond optimally to one another (know as \textit{equilibria}), to be fully described. In SGs, it is assumed that agents have either fixed knowledge of their environment or can acquire knowledge of any missing data by simply observing other agents with which agents can compute best-response actions. However, in many multi-agent systems, agents do not have full information of the environment from the outset and direct computation of optimal behaviour is often prohibitively complex. \newline
\indent Naturally, integrating stochastic game theory with reinforcement learning (RL) - a framework that enables agents to learn optimal behaviour within an unknown environment through direct interaction and exploration - suggests the potential to learn stable policies in multi-agent systems. \newline 
 Although this approach has led to fruitful analysis in multi-agent systems  with few interacting agents \cite{deep2017}, current methods of computing multi-agent equilibria using RL  (e.g. Nash Q-Learning \cite{Well1}, Friend-or-Foe Q-learning \cite{Litt1}, minimax-Q \cite{minimaxQ}) have computational complexity that increases exponentially with the number of agents \cite{busoniu2008comprehensive,tuyls}. This renders the task of using RL to learn equilibrium policies intractable for many systems of interest. \newline  
\indent In this paper, we introduce an approach that enables equilibrium policies of multi-agent systems to be computed even when the size of the population is extremely large. Unlike current multi-agent learning procedures, our method scales independently of the number of interacting agents. In contrast to approaches that compute equilibria in large population games \cite{cardaliaguet}, our method is a model-free, fully decentralised learning procedure that only requires agents to observe local state information and their realised rewards.\newline
\indent Our main result demonstrates that the equilibria of $N-$ player SGs can be computed by solving an optimal control problem (OCP) using a model-free learning procedure under very mild assumptions. To do this, we prove a series of theoretical results: first, we establish a novel link between reinforcement learning in MASs and a class of games known as \emph{discrete-time  mean field games} - $N-$player SGs in an asymptotic regime as $N\to\infty$. Secondly, we demonstrate that in the asymptotic regime, the resultant game belongs to a class of games known as \emph{potential games}. These are reducible to a single objective OCP leading to a vast reduction in the problem complexity. Our last result proves that the equilibria of SGs in an asymptotic regime are in fact approximate equilibria of the $N-$ player SG with an approximation error that vanishes as $N$ increases. Finally, we validate our theoretical results by application to a series of problems within economics and optimal control theory. Our approach is based on a variant of the fictitious play - a belief-based learning rule for static games introduced by Brown \cite{brown1951} and generalised to adaptive play in \cite{Leslie2006} for games with finite action sets.  \newline 
\indent After formulating the problem as an SG, the paper is organised as follows: first, we provide a formal description of a discrete-time mean field game and show that the game is a potential game (Theorem 1). We show that given the potentiality property, the problem is reducible to a single objective OCP. We then introduce the learning protocol and show that under this protocol, the game has strong convergence guarantees to equilibrium policies (Theorem 2). We lastly show that the equilibria generated for the mean field game are approximate equilibria of the $N-$player SG with an approximation error that vanishes as $N$ increases (Theorem 3). Taking benchmark examples from economics and the multi-agent literature, we demonstrate our method within a number of examples with large populations of interacting agents.  
\section{Background}
We now give the background for SG theory by introducing the $N-$player SG formalism. In order to handle multi-agent systems  with large populations, we consider the $N-$player SG formalism when the number of agents tends to infinity - we therefore introduce the notion of mean field games - SGs studied in the asymptotic regime in the number of agents.  
\subsection{Problem formulation: $N-$Player Stochastic Games}
 The canonical framework to describe multi-agent systems  in which agents behave rationally and non-cooperatively is a stochastic (dynamic) game (SG). Let us therefore introduce a formal description of an SG:\newline
\indent Let $\mathcal{N}\triangleq\{1,\ldots,N\}$ denote the set of agents where $N\in\mathbb{N}$. At each time step $k\in 1,2,\ldots T\in\mathbb{N}$,\footnote{The formalism can be straightforwardly extended to infinite horizon cases by appropriate adjustment of the reward function.} the state of agent $i\in\mathcal{N}$ is $x_k^i\in\mathcal{S}^i$ where $\mathcal{S}^i\subset\mathbb{R}^d$ is a $d-$dimensional state space. The state of the system at time $k\leq T$ is given by $x_k\triangleq(x^i_k)_{i\in\mathcal{N}}$ where $x_k\in\mathcal{S}\triangleq\times_{j\in \mathcal{N}}\mathcal{S}^j$, which is the Cartesian product of space of states for each agent. Let $\Pi^i$ be a non-empty compact set of stochastic closed-loop policies\footnote{\emph{Closed-loop} policies are maps from states to actions and are likely to be the only policies that produce optimal behaviour in stochastic systems. \emph{Open-loop} policies simply specify pre-computed (state-independent) sequences of actions. } for agent $i$ where $\pi^i:\mathcal{S}\to\Delta\mathbb{A}^i$ where $\mathbb{A}^i\subseteq\mathbb{R}^d$ is a compact, non-empty action set for each agent $i\in\mathcal{N}$. Denote by $\mathbf{\Pi}$ the set of policies for all agents i.e. $\mathbf{\Pi}\triangleq\times_{j\in \mathcal{N}}\Pi^j$. We denote by $\Pi^{-i}\triangleq\times_{j\in \mathcal{N}\backslash\{i\}}\Pi^j$, the Cartesian product of the policy sets for all agents except agent $i\in\mathcal{N}$. \newline 
 At each time step, each agent $i\in\mathcal{N}$ exercises its policy $\pi^i$, the agent's state then transitions according to the following\footnote{With this specification, agents do not influence each others' transition dynamics directly - this is a natural depiction of various systems e.g. a portfolio manager's modification to their own market position. This does not limit generality since prohibited state transitions (e.g. collisions) can be disallowed with a reward function that heavily penalises such joint action behaviour.}:
\begin{equation}
x^i_{k+1}=f(x^i_k,a_k^i,\zeta_k), \hspace{2 mm} a_k^i \sim \pi^i,\hspace{2 mm} k = 0,1,\ldots,T \label{statetransitionnonMFG}
\end{equation}
where $\{\zeta_k\}_{0\leq k \leq T}$ is a collection of i.i.d. random variables that introduce randomness in the agent's state transition.\newline
 Each agent $i$ has a cumulative reward function $J^i:\mathcal{S}\times\Pi^i\times\Pi^{-i}\to\mathbb{R}$ that it seeks to maximise given by the following: 
\begin{equation}
J^i[x_t,\pi^i, \pi^{-i}] = \mathbb{E}_{x\sim f}\Bigg[\sum_{k=t}^TL(x_k^i,x_k^{-i},a_k^i)\Big|a_k^i\sim \pi^i\Bigg],\label{finitegamepayoff}
\end{equation}
where $x^i_k\in\mathcal{S}^i$ and $x^{-i}_k\in\mathcal{S}^{-i}$ are the state for agent $i$ and the collection of states for agents $j\in \mathcal{N}\backslash\{i\}$ at time $k\leq T$ respectively and $\pi^i_k$ is the policy for agent $i$. The function $L$ is the \emph{instantaneous reward function} which measures the reward received by the agent at each time step. We refer to the system of equations (\ref{statetransitionnonMFG}) - (\ref{finitegamepayoff}) as game (A). \newline \indent We now formalise the notion of optimality within an SG, in particular, a \emph{Markov-Nash-equilibrium} for the game (A) is the solution concept when every agent plays their best response to the policies of other agents. Formally, we define the notion of equilibrium for this game by the following: 
\begin{definition}The strategy profile $\boldsymbol{\pi}=(\pi^i,\pi^{-i})\in\mathbf{\Pi}$ is said to be a \textit{Markov-Nash equilibrium (M-NE) strategy} if, for any policy for agent $i$, $\pi'^i\in\Pi^i$ and $\forall y\in \mathcal{S}$, we have:
\begin{equation}
J^i[y,\pi^i,\pi^{-i}]\geq J^i[y,\pi'^i,\pi^{-i}], \hspace{2 mm} \forall i \in \mathcal{N}. \label{nasheq0}
\end{equation}
\end{definition} 
 The M-NE condition  identifies strategic configurations in which no agent can improve their rewards by a unilateral deviation from their current strategy.\newline \indent We will later consider approximate solutions to the game (A). In order to formalise the notion of an approximate solution, we introduce $\epsilon-$\emph{Markov-Nash equilibria} ($\epsilon-$M-NE) which extends the concept of M-NE to strategy profiles in which the incentive to deviate never exceeds some fixed constant.  The notion of an $\epsilon-$M-NE can be described using an analogous condition to (\ref{nasheq0}). Formally, the strategy profile $(\pi^i,\pi^{-i})\in\mathbf{\Pi}$ is an $\epsilon-$Markov-Nash equilibrium strategy profile if for a given $\epsilon>0$ and for any individual strategy for agent $i$, $\pi'^i\in\Pi^i$ we have that $\forall y\in \mathcal{S}$:
\begin{equation}
J^i[y,\pi^i,\pi^{-i}]\geq J^i[y,\pi'^i,\pi^{-i}]-\epsilon,\hspace{2 mm} \forall i \in \mathcal{N}.\label{nasheqE}
\end{equation}
 \indent Although in principle, methods within RL such as TD learning can be used to compute the equilibrium policies \cite{sutton}, learning takes place in the product space of the state space and the set of actions across agents, so the problem complexity grows exponentially with the number of agents. \newline 
\indent A second issue facing RL within MASs is the appearance of non-stationarity produced by other adaptive agents. During the learning phase, agents update their policies and thus the way they influence the system. In a non-cooperative MAS with even just a few agents learning independently, the presence of other adaptive agents induces the appearance of a non-stationary environment from the perspective of an individual agent. This in turn may severely impair the agent's own reinforcement learning process and lead to complex and non-convergent dynamics \cite{tuyls}. \newline
\indent With these concerns, we present an alternative approach which involves studying the game (A) in an asymptotic regime as $N\to\infty$. This results in a mean field game - an SG with an infinite population which, as we shall show is both reducible to a single OCP and has M-NE that are $\epsilon-$M-NE for $N-$player SGs where $\epsilon \sim \mathcal{O}(\frac{1}{\sqrt{N}}$).\newline
\indent In order to show the discrete-time mean field game is reducible to single OCP, we demonstrate that they belong to a class of games known as dynamic potential games (PGs). Before considering the mean-field game case, let us formally define a PG in the context of an $N-$player SG:
\begin{definition}\label{potentialdef}
An SG is called a (dynamic) potential game (PG) if for each agent $i\in\mathcal{N}$ and for any given strategy profile $\boldsymbol{\pi}\in\mathbf{\Pi}$ there exists a \textit{potential function} $\Omega:\cdot\times\Pi^i\times\Pi^{-i}\to\mathbb{R}$ that satisfies the following condition $\forall \pi'^i\in\Pi^i:$
\begin{equation}
    J^i[\cdot,(\pi^i,\pi^{-i})]-J^i[\cdot,(\pi'^i,\pi^{-i})]=\Omega[\cdot,(\pi^i,\pi^{-i})]-\Omega[\cdot,(\pi'^i,\pi^{-i})].\label{potentialcondition}
\end{equation}
\end{definition}
\indent A PG has the property that any agent's change in reward produced by a unilateral deviation in their strategy is exactly expressible through a single global function. In PGs - the (Nash) equilibria can be found by solving an OCP \cite{MONDERER1996124}. This is a striking result since obtaining the solution to an OCP is, in general, an easier task than standard methods to obtain  equilibria which rely on fixed point arguments.
\subsection{Mean Field Games}
In this section, we introduce a mean field game (MFG) which is a central framework to our approach. Mean field game theory is a mathematical formalism that handles large-population systems of non-cooperative rational agents. MFGs are formulated as SGs in the form (A) analysed at the asymptotic limit as the number of agents tends to infinity. This formulation enables the collective behaviour of agents to be jointly represented by a probability distribution over the state space \cite{lasry}. \newline 
\indent The MFG formulation results in a description of the agents' optimal behaviour that is compactly characterised by a coupled system of partial differential equations. However, obtaining closed analytic solutions (or even approximations by tractable numerical methods) for the system of equations is often unachievable but for specific cases.\newline
\indent This work offers a solution to this problem; in particular we introduce a learning procedure by which the equilibria of MFG can be learned by adaptive agents. Beginning with the case in which the number of agents is finite, we introduce the following empirical measure which describes the $N$ agents' joint state at time $k$:
\begin{align}
m_{x_k}\triangleq\frac{1}{N}\sum_{i=1}^N\delta_{x^i_k}, \label{empirical measure}
\end{align}
where $x^i_k\in\mathcal{S}$ and $\delta_{x}$ is the Dirac-delta distribution evaluated at the point $x\in\mathcal{S}$.\newline
\indent We now study a game with an infinite number of agents by considering the formalism the $N-$player game (A) in the asymptotic regime as $N\to\infty$ which allows us to treat the ensemble in (\ref{empirical measure}) as being continuously distributed over $\mathcal{S}$. We call this limiting behaviour the \emph{mean field limit} which is an application of the law of large numbers for first order (strategic) interactions in the game (A). We observe that by taking the limit as $N\to\infty$ and using de Finetti's theorem\footnote{Given a sequence of indexed random variables $x_1,x_2,\ldots$ which are invariant under permutations of the index, De Finetti's theorem \cite{FIN1} ensures the existence of the random variable $m_{x_k}$ in (\ref{empirical measure}) in the limit as $N\to\infty$.}, we can replace the empirical measure (\ref{empirical measure}) with a probability distribution $m\in\mathbb{P}(\mathcal{H})$ where $\mathbb{P}(\mathcal{H})$ is a space of probability measures. The distribution $m$ describes the joint locations of all agents in terms of a distribution. \newline 
\indent With this structure, instead of agents responding to the actions of other agents individually, each agent now performs its actions in response to the mass which jointly represents the collection of states for all agents. \newline 
\indent As is standard within the MFG framework \cite{lasry}, we assume that the MFG satisfies the \textit{indistinguishably property} - that is the game is invariant under permutation of the agents' indices.\newline
\indent The following concept will allow us to restrict our attention to games with a single M-NE:
\begin{definition}
The function $v:\mathbb{P}(\mathcal{H})\times\cdot\to\mathbb{R}$ is said to be \emph{strictly monotone} in the $L^2-$ norm given $m_1,m_2\in\mathbb{P}(\mathcal{H})$ if the following is satisfied:
\begin{equation*}
\int_\mathcal{S}(v(m_1,\cdot)-v(m_2,\cdot))(m_1-m_2)dx\geq 0 \implies m_1\equiv m_2.\label{monotone}
\end{equation*}
\end{definition}
 The strict monotonicity condition means that in any given state, agents prefer a lower concentration of neighbouring agents. This property is a natural feature within many practical applications in which the presence of others reduces the available rewards for a given agent e.g. spectrum sharing \cite{ahmad2010spectrum}. \newline 
 We make use of the following result which is proved in \cite{lasry}\footnote{In \cite{lasry} the result is proven for mean field games with continuous action and state spaces in continuous time. The corresponding results for discrete games (discrete state space, time and action set) have also been proven (see theorem 2, pg 6 in \cite{GOMES2010}).}:
\begin{proposition}[Lasry \& Lions, 2007]
If the instantaneous reward function of a MFG is strictly monotone in $m\in\mathbb{P}(\mathcal{H})$, then there exists a unique M-NE for the MFG.
\end{proposition}
\section{MAS with Infinite Agents}
 To develop a learning procedure that scales with the number of agents, we now consider the game (A) in the mean-field limit. We shall demonstrate that this procedure allows us to reduce the game (A) to a strategic interaction between an agent and an entity that represents the collection of other agents. This plays a key role in reducing the problem complexity and collapsing it to a single OCP. \newline
\indent We shall firstly define the $N-$player stochastic game (A) in the asymptotic regime. We note that in light of the \emph{indistingushability criterion}, we can drop the agent indices:
\begin{definition}[Discrete-Time Mean Field Game]\label{B} We call the following system a \textit{discrete-time mean field game} if its dynamics can be represented by the following system:
\begin{align}
x_{k+1}&=f(x_k,a_k,\zeta_k) \label{statetransitionMFG}\\
m_{x_{k+1}}&=g(m_{x_k},\mathbf{a}_{k})\label{mtransitionMFG},
\end{align}
\end{definition}
\noindent where $a_k\sim\pi$ for some $\pi\in\Pi$ and $\mathbf{a}_k=(a_k)_{i\in\mathcal{N}}$, $k\in 0,1,\ldots, T$, for a given time horizon $T<\infty$, $m_{x_k}\in\mathbb{P}(\mathcal{H})$ is the agent density corresponding to the asymptotic distribution (\ref{empirical measure}) evaluated at $x_k\in\mathcal{S}$, $\pi_k$  is a policy exercised at each time step $k\leq T$ and $\zeta_k$ is an i.i.d. variable which captures the system stochasticity. We refer to the system (\ref{statetransitionMFG}) - (\ref{mtransitionMFG}) as game (B).\newline
Given $m_{x_t}$ and $\pi\in\Pi$ with the agent index removed, we consider games where each agent has the following reward function:
\begin{equation}
J[x_t,\pi,m_{x_t}]=\mathbb{E}\Big[\sum_{k=t}^TL(x_k,m_{x_k},a_k)\Big|a_k\sim \pi\Big],    \label{payofffunctionmfg}
\end{equation}
where $x_t\in\mathcal{S}$ is some initial state and $a_k$ is the action taken by the agent. \newline 
\indent We are now in a position to describe a MFG system at equilibrium i.e. when each agent plays a best-response to the actions of other agents.\newline 
 \indent Given some initial state $x_t\in\mathcal{S}$, the joint solution $(\tilde{\pi},\tilde{m})$ to the game (B) is described by the following triplet of equations which describes the M-NE: 
\begin{align}
\tilde{\pi}&\in \arg\hspace{-0.65 mm}\max_{\hspace{-1.65 mm}\pi\in\Pi}J[x_t,\pi,\tilde{m}_{x_t}],\label{controlMFG*}\\
\tilde{x}_{k+1}&=f(\tilde{x}_k,\tilde{a}_k,\zeta_k), \hspace{1.8 mm}
\tilde{{a}}_{k}\sim\tilde{{\pi}}\in\Pi\\\tilde{m}_{\tilde{x}_{k+1}}&=g(\tilde{m}_{\tilde{x}_k},\tilde{\mathbf{a}}_{k}),\hspace{1.8 mm} \tilde{\mathbf{a}}_{k}=(\tilde{a}_k)_{i\in\mathcal{N}}\label{mtransitionMFG*},
\end{align}
where as before, $\{\zeta_k\}_{0\leq k \leq T}$ is a collection of i.i.d random variables and, $m_{\tilde{x}}$ is the agent density induced when the policy $\tilde{\pi}$ is exercised by each agent. \newline
\indent An important feature of the system (\ref{controlMFG*}) - (\ref{mtransitionMFG*}) is that the agent's problem is reduced to a strategic interaction between itself and a single entity $\tilde{m}_x$. This property serves a crucial role in overcoming the appearance of non-stationarity in an environment with many adaptive learners since the influence of all other agents on the system is now fully captured a single entity $\tilde{m}_x$ which, influences the system dynamics in a way that an adaptive agent can learn its optimal policy.  \newline
\indent Existing methods of computing equilibria in MFGs however rely on the agents having full knowledge of the environment to compute best responses and involve solving to non-linear partial differential equations \cite{cardaliaguet,master2} which, in a number of cases leads to intractability of the framework. MFGs are closely related to anonymous games - games in which the agents' rewards do not depend on the identity of the agents they interact with (but do depend on the interacting agents' strategies). Multi-agent learning has been studied for anonymous games \cite{kash} however, this approach requires agents to fix their policies over \emph{stages} and to explicitly compute approximate best-responses. In the following sections of the paper, we develop a technique which enables equilibrium policies of MFGs to be computed by adaptive learners in an unknown environment without solving partial differential equations. 
\section{Theoretical Contribution}
\subsection{Mean Field Games are Potential Games}
$\indent$ We now demonstrate that the discrete-time MFG problem (B) is reducible to an objective maximisation problem. By proving that the discrete-time mean field game is a PG, the following theorem enables us to reduce the problem to a single OCP: 
\begin{theorem}\label{DTMFGsarepotential}
\textit{ The discrete-time mean field game (B) is a PG.}
\end{theorem}
We defer the proof of the theorem to the appendix. The key insight of Theorem \ref{DTMFGsarepotential} is that the M-NE of MFGs can be computed by considering a general form of a \textit{team game} in which each agent seeks to maximise  the potential function. Crucially, thanks to Theorem \ref{DTMFGsarepotential}, the problem of computing the equilibrium policy is reduced to solving a control problem for the potential function. \subsection{Learning in Large Population MAS}
$\indent$ We now develop a model-free decentralised learning procedure based on a variant of fictitious play using the potentiality property. This generates a sequence of polices that converges to the M-NE of the discrete-time MFG.    \newline
 \indent Firstly, it is necessary to introduce some concepts relating to convergence to equilibria:
\begin{definition}\label{path} 
Let $\{\pi^{i,n}\}_{n\geq 1}$ be a set of policies for agent $i\in\mathcal{N}$. We define a \textit{path} by a sequence of strategies $\rho_\pi^i\triangleq(\pi^{i,n})_{n\geq 1}\in\times_{n\geq 1}\Pi^i$, where $\pi^{i,n+1}$ is obtained from an update of $\pi^{i,n}$ using some given learning rule.   
\end{definition} 
\begin{definition}\label{improvement path}
Given $\pi^{-i}\in\Pi^{-i}$, the path $\rho_\pi^i\in\times_{n\geq 1}\Pi^i$ is called an \textit{improvement path} for agent $i$ if after every update the agent's expected reward increases, formally  an improvement path satisfies the following condition:
\begin{equation}
J^i[\cdot,\pi^{i,n+1},\pi^{-i}]\geq J^i[\cdot,\pi^{i,n},\pi^{-i}],\hspace{3 mm} \forall i \in \mathcal{N}. \label{improvementpath}
\end{equation}
\end{definition}
\begin{definition}
 A path converges to equilibrium if each limit point is an equilibrium.
\end{definition}
We now describe a \lq belief-based\rq$ $ learning rule known as fictitious play \cite{brown1951} of which our method is a variant: Let $\rho^i_\pi\in\times_{n\geq 1}\Pi^i$ be a path, then the learning rule is a \emph{fictitious play process} (FPP) if the update in the sequence $\{\pi^{i,n}\}_{n\geq 1}$ is performed $\forall x\in\mathcal{S}$,  $\forall i\in\mathcal{N}$ as follows:\newline
\begin{equation}
\sup_{\pi'\in\Pi}J^i[x,\pi',\pi^{-i,n}]=J^i[x,\pi^{i,n+1},\pi^{-i,n}],    
\end{equation}
so that $\pi^{i,n+1}$ is a best-response policy against $\pi^{-i,n}$.\newline
If the FPP converges to equilibrium then we say that the game has the \textit{fictitious play property}. \newline
\indent We now apply these definitions to the case of MFGs. We note that by the indistinguishability assumption for MFGs, we can drop the agent indices in each of the above definitions. We define the FPP for the MFG (B) by the following learning procedure $\forall x\in\mathcal{S}$:
\begin{gather*}
\hspace{-1.6 mm}\sup_{\pi'\in\Pi}J[x,\pi',\bar{m}^n_{x}]=J[x,\pi^{n+1},\bar{m}^n_{x}],\hspace{1.8 mm}
\bar{m}^n_{x}\triangleq\frac{1}{n}\sum_{j=1}^nm^j_{x}
\end{gather*} so that $\pi^{n+1}$ is a best-response policy against $\bar{m}^n_{x}$ which summarises each agent's belief of the joint state of all agents after the $n^{th}$ update.\newline
 In order to solve game (B) we therefore seek a learning process that produces a sequence $\{(\pi^n,m^n)\}_{n\geq 1}$ for which $\{\pi^n\}_{n\geq 1}$ is an improvement path for the policy $\pi^n$.\newline 
\indent  For the discrete-time MFG, we shall seek a pair $(\pi^n,m^n)\in\Pi\times\mathbb{P}(\mathcal{H})$ that converges to equilibrium as $n\to\infty$ so that the sequence $\{(\pi^n,m^n)\}_{n\in\mathbb{N}}$ converges to a cluster point $(\tilde{\pi},\tilde{m})$ which is a solution to (B).\newline
\indent Our convergence result is constructed using results that we now establish. Before proving the result we report an important result in the model-based setting: 
\begin{proposition}[Cardaliaguet, Hadikhanloo; 2017]Mean Field Games have the fictitious play property.\end{proposition}
This result was established in \cite{cardaliaguet} within a continuous-time and model-based setting. Here, given some initial belief about the distribution $m^n$ and some initial value function $v^n$ associated to each agent's problem, the agents update the pair $(v^n,m^n)$ according to a (model-based) fictitious play procedure. This produces a paired sequence $\{(v^n,m^n)\}_{n\geq 1}$ for which $\lim_{n\to\infty} (v^n,m^n)=(\tilde{v},\tilde{m})$ where $(\tilde{v},\tilde{m})$ is joint solution to the continuous-time MFG.\newline
\indent In order to compute the best responses at each step, the FPP discussed in \cite{cardaliaguet} requires agents to use knowledge of their reward functions. Moreover, the agents' update procedure involves solving a system of partial differential equations at each time step. Obtaining closed solutions to this system of equations is generally an extremely difficult task and often no method of obtaining closed solutions exists.  \newline 
\indent We are therefore interested in procedures for which agents can achieve their M-NE policies by simple adaptive play with no prior knowledge of the environment. 
 We are now in position to state our main result:
\begin{theorem}\label{fictitousplaytheorem}
 There exists a fictitious-play improvement path process such that the sequence $\{(\pi^n,m^n)\}_{n\geq 1}\in\Pi\times\mathbb{P}(\mathcal{H})$ converges to an $\epsilon-$M-NE of the game (A).
\end{theorem}
 The following corollary demonstrates that we can construct a learning procedure that leads to an improvement path, the limit point of which is a solution of the MFG (B). \begin{corollary}\label{convergencepropostion}
Let $\{(u^n,m^n)\}_{n\geq 1}$ be a mean field improvement path generated by an actor-critic fictitious play method, then $\{(u^n,m^n)\}$ converges to a cluster point $\{(u,m)\}$. Moreover, the cluster point $\{(u,m)\}$ is a solution to (B).
\end{corollary}
 Corollary \ref{convergencepropostion} immediately leads to our method which computes the optimal policies for the MFG (B). The method uses an actor-critic framework with TD learning on the critic and policy gradient on the learner. An episode is simulated using some initial belief over the distribution $m$ over $\mathcal{S}$. The agent then updates its policy using an actor-critic and updates the distribution $m^k$ fictitiously.
\subsection{The Approximation Error}
$\indent$In this section, we show that the Nash equilibria generated by the game (B) are approximate equilibria for the $N-$player stochastic game (A). 
\begin{theorem}\label{approximationtheorem}
 Let $\bar{\pi},\tilde{\pi}\in{\Pi}$ be the NE strategy profile for the game (A) and game (B) respectively and let  be a NE strategy profile for the MFG (B). Let $\tilde{m}_x$ and $\bar{m}_x$ be the distributions generated by the agents in the mean field game and the $N-$player SG respectively; then there exists a constant $c>0$ s.th $\forall x \in  \mathcal{S}$:
 \begin{equation}
|J(x,\tilde{\pi},\tilde{m}_x)-J(x,\bar{\pi},\bar{m}_x)|<c\backslash \sqrt{N}.
\end{equation}
\end{theorem}
 Theorem \ref{approximationtheorem} says that the solution to the MFG (B) is in fact an $\epsilon-$M-NE to the $N-$agent SG (A). Moreover, the approximation error from using a MFG to approximate the  $N-$player game is $\mathcal{O}\big(\frac{1}{\sqrt{N}}\big)$. \newline 
\indent As a direct consequence of theorem \ref{approximationtheorem}, we can use the actor-critic fictitious play method on the MFG formulation (B) to compute near-optimal policy solutions of the stochastic game (A). Moreover, the error produced by the mean field approximation vanishes asymptotically as we consider systems with \emph{increasing} numbers of agents. 
\section{Experiments}
$\indent$ To investigate convergence of our method, we present three experiments drawn from benchmark problems within economics and control theory that involve large populations of strategically interacting agents. In each case, we show that our method converges to an M-NE policy.\newline 
\indent We firstly demonstrate the use of our technique in a (stochastic) congestion game, testing the convergence to a stationary policy in a large-population system with a complicated reward structure. The second problem is a supply and demand problem that we formulate as an SG allowing us to test convergence to optimal policies in dynamic problems requiring long-term strategic planning in the presence of other learning agents. This demonstrates that our method is able to overcome the non-stationary interference of other adaptive agents. Lastly, we apply our method to study a multi-agent generalisation of a fundamental problem within optimal control theory, namely the linear quadratic control (LQC) problem. The analytic solution of the LQC problem allows us to verify that our method converges to a known M-NE policy.\\
\textbf{Experiment 1:} \textit{Spatial Congestion Game}\\
 \indent In the spatial congestion game the rewards are dependent on the agents' use of a shared resource (a sub-region of $\mathcal{S}\subset \mathbb{R}^2$) and the number agents using that resource. Games of this type are known as \emph{congestion games} and represent a large class of interactions e.g. spectrum sharing problems.\newline \indent
In the spatial congestion (SC) game there are $N$ agents, given some initial position $x_0\in\mathcal{S}$, each agent chooses an action in order to move to a desired location $x_T\in\mathcal{S}$ which is a terminal state. Certain areas of $\mathcal{S}$ are more desirable to occupy than others, however the agents are averse to occupying crowded areas - they receive the largest rewards for occupying parts of $\mathcal{S}$ that are both desirable and have relatively low concentrations of agents. The agents simultaneously select a movement vector $u\in\mathbb{R}^{2\times 1}$ resulting in movement to a terminal state $x_T$. Each agent then receives its reward $L$ which depends on the desirability of the location and the concentration of agents $m_{x_T}$ at $x_T$.\newline 
\indent Formally, we model the desirability of a region $x_t \in \mathcal{S}$ at time $t$ as\footnote{We note that the function $L$ is continuously differentiable in $\mathcal{S}$ so that assumption 2 is \textit{a fortiori} satisfied, moreover, it can be easily verified that assumption 3 holds.}: 
\begin{equation*}
L(x_t,m_{x_t})=\frac{1}{2\pi\sqrt{|\Sigma|}}\frac{e^{-(x_t-\mu)^T\Sigma^{-1}(x_t-\mu)}}{(1+m_{x_t})^{\alpha}}, \label{Gaussianfunction}
\end{equation*} 
where $m_{x_t}\in\mathbb{P}(\mathcal{H})$ is the density of agents at the point $x_t$ and $\mu\in\mathbb{R}^2, \Sigma\propto 1_{2\times2}$ are given parameters representing the mean and spread of the distribution of the rewards over $\mathcal{S}$. The map $L:\mathcal{S}\times \mathbb{P}(\mathcal{H})\to \mathbb{R}$ measures the instantaneous reward for an agent at $x_t$ with a local agent density $m_{x_t}$. The parameter $\alpha>0$ is a measure of each agent's averseness to occupying the same region as other agents with higher values representing greater averseness. \newline 
Given some initial position $x_0 \in\mathcal{S}$ and $u\in\mathbb{R}^{2\times 1}$, the transition dynamics are given by the following expression:
\begin{equation}
x_T=f_1(x_0,u,\epsilon)\triangleq A_1x_0+B_1u+\sigma_1\epsilon,
\end{equation}
where $\epsilon \sim \mathcal{N}(0,\sigma_\epsilon);\sigma_1, A_1, B_1\propto  1_{(2\times 2)}$ and  $c,\sigma_\epsilon\in\mathbb{R}^+$.
\newline 
The reward function for an agent is then given by the following expression:
\begin{equation*}
\hspace{-1 mm} J[x_0,\pi,m_{x_0}]=\mathbb{E}_{x_T\sim f_1}\Big[L(x_T,m_{x_T})-\frac{1}{2}u^TRu\big| u\sim \pi\Big],    \label{rewardfunctionos1}
\end{equation*}
where $R=\eta 1_{(2\times 2)}$ is a control weight matrix and $\eta\in\mathbb{R}$ is the marginal control cost (cost of movement). Using the indistingishability condition, we have omitted agent indices.\newline
At equilibrium each agent optimally trades-off state-dependent rewards with its proximity to nearby agents.\newline
 \indent The problem generalises the beach domain problem studied in \cite{devlin} since we now consider a reward function with state dependency. In particular, the desirability over $\mathcal{S}$ is described by a Gaussian function over $\mathcal{S}$. Moreover the problem we now consider consists of a system with 1000 interacting agents. The problem is also closely related to the spectrum sharing problem, see \cite{ahmad2010spectrum}.
 \begin{figure}
    \centering
    \includegraphics[width=0.5\columnwidth]{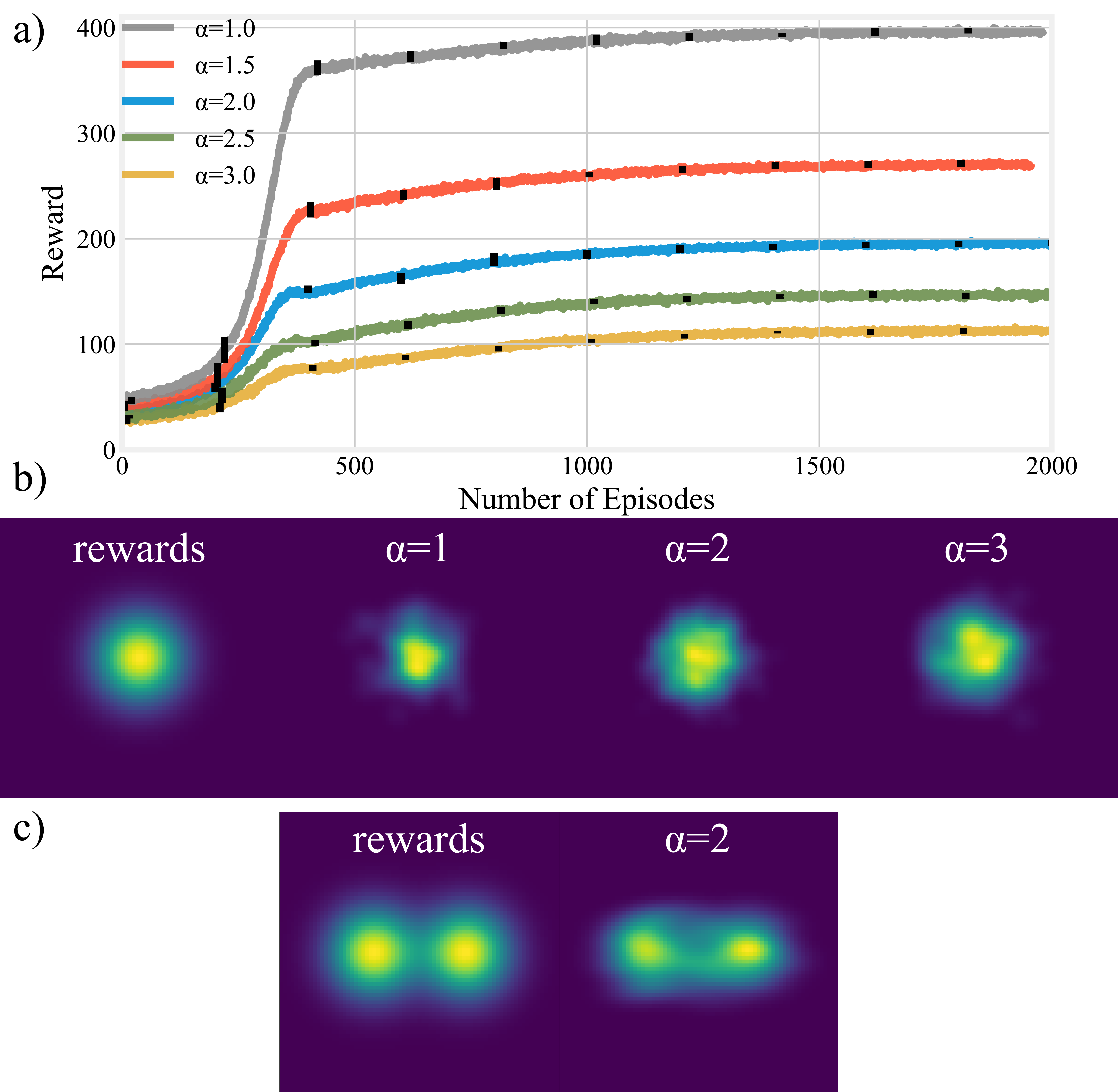}
    \caption{a) Rewards over 2000 episodes of training for a Gaussian distribution  of desirability. b) agent distributions for different averseness parameter $\alpha$. c) agent distribution for bimodal Gaussian function reward function.}
    \label{fig:my_label}
\end{figure}
\indent In accordance with the theory, our method converges to a stable policy - after 2,000 episodes of training we find that the agent's policy and rewards stabilise Figure 1.a).  Figure 1.b) shows the terminal distribution of agents over $\mathcal{S}$ for $\alpha\in\{1.0, 2.0, 3.0\}$. We observe that the agents learn to optimally trade-off state-dependent rewards with distance from neighbouring agents resulting in a fixed terminal distribution of agents. As expected, the agents disperse themselves further as the value of $\alpha$ is increased - in all cases converging to a stable distribution over $\mathcal{S}$.  \\
\indent Figure 1.c) shows the distribution of agents for a more complicated reward structure specified by a mixture of two Gaussians over $\mathcal{S}$  with peaks at $(-1,0)$ and $(0,0)$). We initialise the agents at the point $(1,0)$. In this case, an individual agent's reward is given by the following expression:
\begin{equation*}
\hspace{-1.5 mm}J[x_0,\pi,m_{x_0}]=\mathbb{E}_{x_T\sim f_1}\Big[\sum_{i=1}^2L_i(x_T,m_{x_T})-\frac{1}{4}u_i^TRu_i\Big]
\end{equation*}
where $L_i(x_t,m_{x_t})\triangleq[16\pi^2|\Sigma_i|]^{-\frac{1}{2}}e^{-(x_t-\mu_i)^T\Sigma_i^{-1}(x_t-\mu_i)}\\\cdot(1+m_{x_t})^{-\alpha}$ 
and $u_i\sim\pi$.
Since learning is internal to each agent, a possible (suboptimal) outcome is for the agents to cluster at the nearest peak of rewards. However, using our method, the agents learn to spread themselves across the state space and distribute themselves across both peaks.\newline
\textbf{Experiment 2:} \textit{Supply with Uncertain Demand}\\ Optimally distributing goods and services according to demand is a fundamental problem within logistics and industrial organisation. In order to maximise their revenue, firms must strategically locate their supplies given some uncertain future demand whilst considering the actions of rival firms which may reduce the firm's own prospects.\newline 
 \indent We now apply our method to a supply and demand problem in which individual firms seek to maximise their revenue by strategically placing their goods when the demand process has future uncertainty. The demand process, which quantifies the level of demand associated with each point in space, is \emph{a priori} unknown and is affected by the actions of thousands of rival firms. Each firm directs supply of its goods to regions in time and space however, the firms face transport costs so that each firm seeks to optimally trade-off transportation costs and tracking the demand. As firms begin to concentrate on a particular area of demand, the sale opportunities  diminish, reducing the rewards associated to that region of demand. \newline 
\indent We model this problem as an episodic problem with a distribution of rewards traversing a path through the state space $\mathcal{S}\in\mathbb{R}^2$ (illustrated in Figure 2.a)). Agents seek to locate themselves in areas of high concentrations of rewards for a fixed number (30) time steps. The agents are penalised for both movement and occupying areas with a high density of other agents. The experiment tests the ability of the method to avoid convergence to suboptimal outcomes. In particular, movement costs are highly convex so traversing the path of rewards leads to low overall rewards. \begin{figure}
    \centering
    \includegraphics[width=0.5\columnwidth]{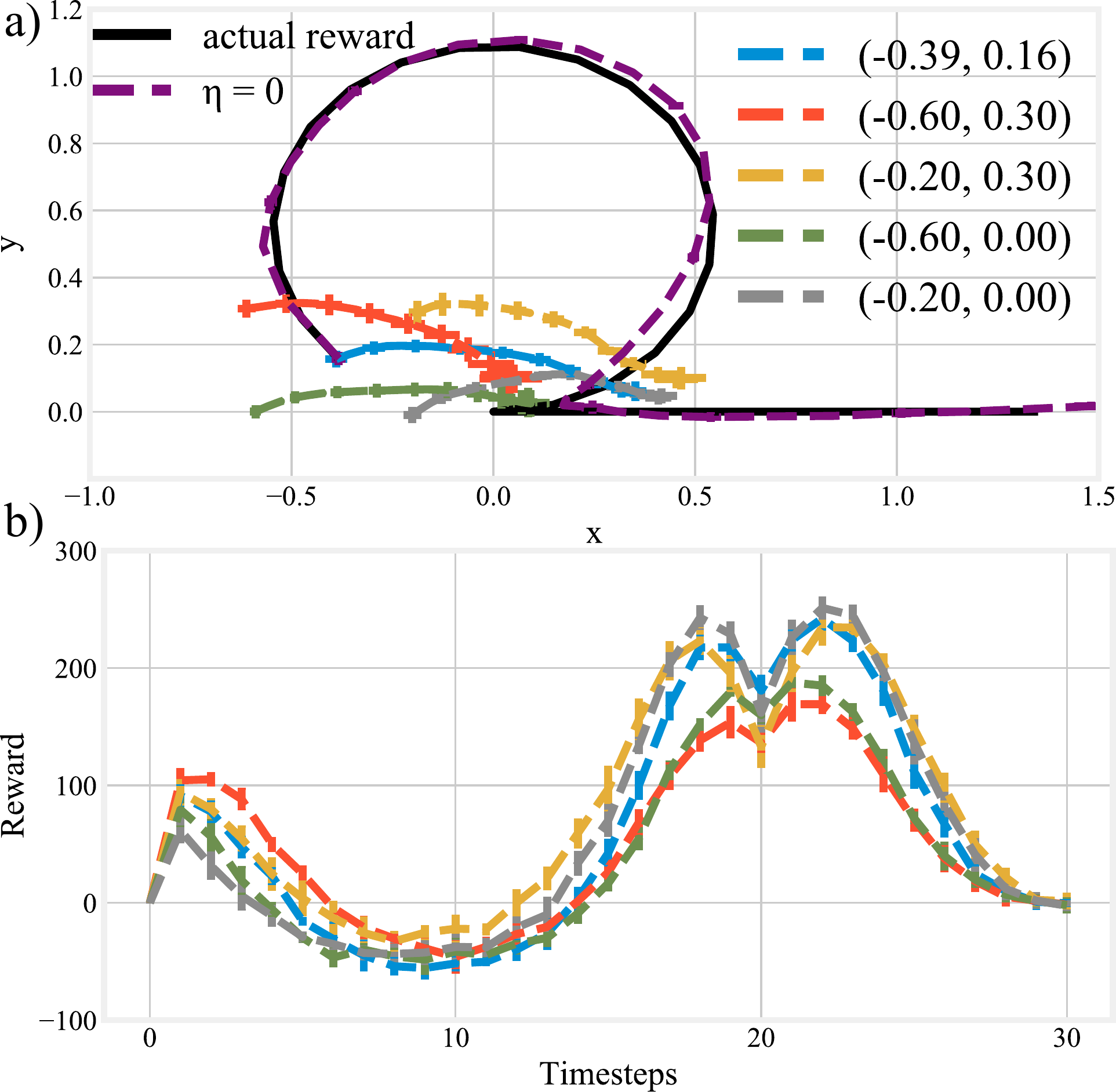}
    \caption{a) Paths of agents with different initial points (coordinates indicated) for cost of movement $\eta=2$. Shown in black is the path of the rewards (demand) and in purple is the path of agents with $\eta=0$. b) The intra-episode rewards for the paths show in a). Note incursion of negative rewards.}
    \label{fig:my_label2}
\end{figure}\\
\indent Figure 2.a) illustrates the path of the agents after training for several different initial positions. An interesting outcome is that with non-zero movement costs, the agents learn to move horizontally to intersect the path of the rewards at a later time. This behaviour conforms with intuition - to maximise  long-term rewards the agents must choose a path that initially incurs higher costs, forgoing immediate rewards whilst they traverse the regions with sparse rewards. \newline 
Figure 2.b) shows that the agents are able to learn to initially incur negative rewards to maximise their cumulative payoffs. Without long-term planning each agent would attempt to trace the same path as the rewards. Such a strategy would lead to reduced overall rewards since attempting to match the locations of the rewards is prohibitively costly. Setting $\eta=0$ we see that the agents trace the path of the rewards (represented by the dashed purple line in Figure 2.a)). \\
\textbf{Experiment 3:} \textit{Mean-Field Linear Quadratic Regulator}\\
 The linear quadratic control (LQC) problem is a fundamental  problem within optimal stochastic control theory (OSCT). It concerns a system whose transition dynamics evolves according to a stochastic process that is linearly controlled subject to quadratic costs. The LQC problem essentially captures the local problem of a large class of problems in OSCT and can be solved analytically, the solution being given by the linear-quadratic regulator \cite{bardi,Xu2015}. LQC models have been extended to mean field interactions in which a large population of agents affect the dynamics of a system using linear controls subject to quadratic control costs and a cost term which depends on the actions of other agents.\newline 
 \indent The reward function for the LQC problem is given by:
\begin{equation*}
J[x_0,u,m_{x_0}]=\mathbb{E}_{x_T\sim f_2}\Big[\sum_{t=0}^{T_n}\{C(x_t,\bar{m}_{x_t})-\frac{1}{2}u_t^TRu_t\}\Big],
\end{equation*}
where $C(x_t,m_{x_t})\triangleq-(x_t-\alpha)^TQ_t(x_t-\alpha)$.\newline 
 At time $k\leq T_n$, given some position $x_k \in\mathcal{S}$, each agent then chooses an vector control parameter $u_t\in\mathbb{R}^{2\times 1}$. The transition are given by the following expression:
\begin{equation}
x_{k+1}=f_2(x_k,u_k,\epsilon_k)\triangleq A_1x_k+B_1u_k+\sigma_1\epsilon_k,
\end{equation}
where $\epsilon_k \sim \mathcal{N}(0,\sigma_{\epsilon_k}) \hspace{1 mm}
\forall k< T_n,  A_1,B_1\propto  1_{(2\times 2)}$, $\sigma_1 = c1_{(2\times 2)}$ where $c\in\mathbb{R}^+$ is some constant that measures the magnitude of the stochasticity in each agent's transition.\\
\indent In \cite{bardi} the distribution $m_x$ is reported after convergence to the stationary M-NE. We compare our results (E3) with this stationary policy (B1) in the following table:
\begin{center}
 \begin{tabular}{||c| c| c ||} 
 \hline
 & B1 & E3 \\ [0.5ex] 
 \hline
 $\mu$ & (0.50000, -0.50000) & (0.50717, -0.50537) \\ 
 \hline
 $\sigma^2$ & 0.14060 & 0.16100  \\  \hline
\end{tabular}
\end{center}
Clearly, our results converge to values that closely replicate the analytic solution.
\section{Conclusion}
We develop an approach to MARL with large numbers of agents. This is the first paper to prove convergence results to best-response policies in multi-agent systems with an unbounded number of agents. This allows RL to be applied across a broader range of applications with large agent populations, in contrast to current methods \cite{Well1,Litt1,minimaxQ}. Our approach advances existing work in MFGs as in \cite{cardaliaguet} that require both full knowledge of the environment and to perform involved analytic computation. In contrast, by developing a connection between RL in MASs and MFGs, we demonstrate a procedure that is model-free, enabling agents to learn best-response policies solely through adaptive play which overcomes the problem of non-stationarity. In our experiments we provide a novel approach of analysing problems in control theory and economics.    
\section{Acknowledgements}
$\indent$ We would like to thank Haitham Bou-Ammar, Sofia Ceppi and Sergio Valcarcel Macua for helpful comments.

\newpage

\appendix
\begin{center}
\section{Supplementary Material}
\end{center}
$\qquad$

\subsection*{Experimental Setup}
$\indent$In all three experiments, all learning was performed in TensorFlow using actor-critic method with the Adam optimiser. The actor is represented by a two-layer fully connected network represented by a neural network with a Gaussian output with variance 0.1. The critic is a two-layer fully connected neural network. The experimental parameters are as follows:
the learning rates in experiments 1, 2 and 3 were set to $10^{-4}$,  $10^{-3}$ and $10^{-3}$ respectively. In all experiments we used a discount factor of 0.99. Experiment 1 is a one-shot game, experiment 2 was run for 30 time steps. We ran experiment 1 with 1,000 agents and experiments 2 and 3 with 200 agents. In experiment 1, the agents' initial position was sampled from a Gaussian distribution $\mathcal{N}((1,0),0.1)$. For experiment 2 the agents' initial position was sampled from a Gaussian $\mathcal{N}((x,0.10)$ where the mean $x$ took the values $x\in\{(-0.20,0.00),(-0.20,0.30),\\(-0.39,0.16),(-0.60,0.00),(-0.6,0.30)\}$. In experiments 1 and 2, for each initial position we performed 6 runs. For experiment 1 we set the marginal cost parameter $\eta=0$ and for experiment 2, $\eta=2$. In experiment 1 we tested for $\alpha\in\{1.0, 1.5,2.0,2.5,3.0\}$ and in experiment 2 we set $\alpha=0.1$.  
\subsection*{Assumptions}
The results within the paper are built under the  following assumptions on the instantaneous reward function $L$ and the functions $f$ and $g$.

\begin{assumption}
The function $f(x,y,\cdot)$ is Lipschitz continuous in $(x,y)$.
\end{assumption}
\begin{assumption}
The function $g(x,y)$ is Lipschitz continuous in $(x,y)$.
\end{assumption}
\begin{assumption}
 The instantaneous reward function $L(s,x,y,u)$ is H\"older-continuous in $(s,x,y)$ and convex in $u$. 
\end{assumption}
\begin{assumption}
 The instantaneous reward function $L$ is monotone in $m$ (in the sense of definition 3) and both convex and separable $u$.
\end{assumption}
\begin{assumption}
The function $L(\cdot,m):\mathbb{P}(\mathcal{H})\to\mathbb{R}$ is $C^1$ in the sense of \cite{master2}, that is there exists a continuous map $\frac{\delta L}{\delta m}$ s.t.:
\begin{equation*}
L(\cdot,m')-L(\cdot,m)=\int\int\frac{\delta L}{\delta m}(x,(1-t)m+tm')(m'-m)(dx)dt \hspace{3 mm}\forall m,m'\in\mathbb{P}(\mathcal{H}).
\end{equation*}
\end{assumption}
\begin{assumption}
 The instantaneous reward function $L$ is bounded.
\end{assumption}

\subsection*{Technical Proofs}
\subsubsection*{Proof of Theorem 1}
The following set of results are instructive for the proof of theorem 1:
 \begin{proposition}
Identical interest games are PGs.
\end{proposition}
The following result allows us to make use of proposition 1 in the context of mean field games:\begin{lemma}
Mean field games are identical interest games.\end{lemma}
\begin{proof}
The result follows immediately from the indistinguishability condition - in particular we note that the potentiality condition (c.f. in Definition 2)) is satisfied by the instantaneous function $L$ itself.\end{proof} 
To prove the theorem, we exploit directly assumption 4 from which we can deduce the existence of the quantity $\frac{\delta L}{\delta m}$. Thus, in full analogy with lemma 4.4 in \cite{MONDERER1996124}, we conclude that the game has a potential and that the equilibria of the mean field game (B) can be obtained by maximising the potential function. 

\subsubsection*{Proof of Theorem 3}
\begin{proof}
The proof of theorem 3 exploits both the boundedness properties and continuity of the functions $f$ and $J$. We build the proof in two parts the main part of which is given by the following result which we shall prove immediately: \newline 
\indent Given NE strategy profile for the game (A), $\bar{\pi}\in\Pi$ and a NE strategy profile for the mean field game (B) $\tilde{\pi}$ then the following inequalities hold:
\begin{enumerate}
\item $|f(\tilde{x},\cdot)-f(\bar{x},\cdot)|<\frac{c}{\sqrt{N}}$
\item $|J(\tilde{x},\tilde{\pi},\tilde{m}_{\tilde{x}})-J(\bar{x},\bar{\pi},\bar{m}_{x})|< c\sup_{k\in[t,T]}|\tilde{x}_k-\bar{x}_k|$
\end{enumerate}
where $c>0$ is an arbitrary constant (that may vary in each line) and as before $\tilde{m}$ and $\bar{m}$ are the distributions generated by the agents in the mean field game and the $N-$player stochastic game respectively. \newline 
\indent The first inequality bounds the difference in trajectories between the instance that the agents use the control for the mean field game and the control for the $N-$player stochastic game. The second inequality bounds the change in rewards received by the agents after perturbations in their location. Since the agents use closed-loop policies, we will make use of the result to describe changes in their policies.\newline 
\indent Let $N$ be the number of agents in the $N-$player stochastic game (A), then given $\bar{\pi}$ and $\tilde{\pi}$ as in part I, if $x^N$ and $\tilde{x}$ are solutions to the processes (1) and (7) respectively, i.e. $x^N_{k+1}=f(x^i_k,a^i_k,\epsilon_k^i)$ and $\tilde{x}_{k+1}=f(x_k,a_k,\epsilon_k)$ where $a^i_k\sim\bar{\pi}$ and $\tilde{a}_k\sim\tilde{\pi}$.\newline \indent To prove $ \|x^N_k-\tilde{x}_k\|\leq \frac{c_k}{\sqrt{N}}$ we firstly must modify the representation of the function $f$ as a function purely in terms of the spatial variables.\newline
\indent By \cite{Xu2015}, the policy $\pi$ can be expressed by a function of the form $\Gamma:\mathcal{S}_i\times\mathcal{S}_{-i}\to\mathbb{R}$ s.t.: $\pi=\Gamma(x^i,x^{-i})$, moreover, the function $\Gamma$ is Lipschitz continuous in each variable.\newline
\indent We note that for any given pair $x_k^i\in\mathcal{S}^i,x_k^{-i}\in\mathcal{S}^{-i}$ we can therefore express the transition function $f$ in the following way $\forall k\leq T$:
\begin{equation}
f(x_k^i,a^i_k,\zeta_k)\equiv\Theta(x_k^i,x^{-i}_k,\zeta_k), \hspace{4 mm} a^i_k\sim \pi\tag{A\theequation}\label{fistheta}\stepcounter{equation}\end{equation}
where $\{\zeta_k\}$ is a set of i.i.d. random variables and the function $\Theta$ is a bounded, Lipschitz continuous function. \newline \indent By the same reasoning (and recalling (6) - (7)), given a NE policy for the discrete MFG (B), $\tilde{\pi}\in \Pi$ and the induced mean field distribution $m_{\tilde{x}}\in\mathbb{P}(\mathcal{H})$,  we can similarly express the transition dynamics for the discrete-time MFG in terms of the function $\Theta$ as $f(\tilde{x}_k,a_k,\zeta_k)\equiv\Theta(\tilde{x}_k,m_{\tilde{x}_k},\zeta_k)$ where $a_k\sim\tilde{\pi}$.
Moreover, given the boundedness of the function $\Theta$, we also make the following observation which bounds the variance taken w.r.t $x\in\mathcal{S}$, $\forall k\leq T$:
\begin{equation}
\sum_{l=1}^d {\rm var}\Bigg(\frac{1}{n}\sum_{j\neq i}^N\Theta(x^i_k,x^{-i}_k,\cdot)\Bigg)\leq \frac{k}{n}(1+\max_{l}\|\Theta_l\|_\infty)^2\leq \frac{c_k}{N},\tag{A\theequation}\label{variancebound}\stepcounter{equation}
\end{equation}
where $\|f_l\|$ denotes the bound on the $l^{th}$ component of $\Theta$.   
\newline 
\indent We secondly note that using (\ref{fistheta}) and by the state transition equations (1) and (7) we note that we can express the difference in trajectories as:
\begin{equation*}
\|x^N_k-\tilde{x}_k\|=\mathbb{E}\Big[\Big\|\frac{1}{N}\sum_{j\neq i}^N\Theta(\bar{x}_k^i,\bar{x}^j_k,\zeta)-\int_{\mathcal{S}}\Theta(\tilde{x}_k,x',\zeta)m(dx')\Big\|\Big].
\end{equation*}
We now prove the statement by induction on the time index $k\in\mathbb{N}$. 
\newline
\indent In the base case, we seek to prove the following bound:
\begin{equation*}
\|x^N_1-\tilde{x}_1\|\leq \frac{c_k}{\sqrt{N}},
\end{equation*}
which is equivalent to the following inequality:
\begin{equation*}
\mathbb{E}\Big[\Big\|\frac{1}{N}\sum_{j\neq i}^N\Theta(x_0,x^j_0,\zeta)-\int_{\mathcal{S}}\Theta(x_0,x',\cdot)m(dx')\Big\|\Big]\leq \frac{c_k}{\sqrt{N}}    
\end{equation*}
This follows straightforwardly since 
\begin{gather*}
\mathbb{E}\Big[\Big\|\frac{1}{N}\sum_{j\neq i}^N\Theta(x_0,x^j_0,\zeta)-\int_{\mathcal{S}}\Theta(x_0,x',\cdot)m(dx')\Big\|^2\Big]\nonumber\\\leq
\sum_{l=1}^d{\rm var}\Bigg(\frac{1}{n}\sum_{j\neq i}^N\Theta(x_0,x_{-i},\cdot)\Bigg)\leq \frac{c_k}{N},
\end{gather*}
using the variance bound (\ref{variancebound}) by reordering the expression, after applying the Cauchy-Schwarz inequality we arrive at the required result.\newline 
\indent For the general case, we firstly make the following inductive hypothesis:
\begin{equation*}
\|x^N_k-\tilde{x}_k\|\leq \frac{c_k}{\sqrt{N}},
\end{equation*}
We therefore seek to show that the following bound is satisfied:
\begin{equation*}
\|x^N_{k+1}-\tilde{x}_{k+1}\|\leq \frac{c_k}{\sqrt{N}},
\end{equation*}
which is equivalent to:
\begin{equation}
\mathbb{E}\Big[\Big\|\frac{1}{N}\sum_{j\neq i}^N\Theta(\bar{x}_k^i,\bar{x}^j_k,\zeta)-\int_{\mathcal{S}}\Theta(\tilde{x}_k,x',\zeta)m(dx')\Big\|\Big]\leq \frac{c_k}{\sqrt{N}}. \tag{A\theequation}\label{indhyp2}\stepcounter{equation}   
\end{equation}
To achieve this, we consider the term:
\begin{equation*}
\mathbb{E}\Big[\Big\|\frac{1}{N}\sum_{j\neq i}^N\Theta(x_k^i,x^j_k,\zeta)-\int_{\mathcal{S}}\Theta(x_k^i,x',\zeta)m(dx')\Big\|\Big].\end{equation*}
Now by the triangle inequality and the Lipschitz continuity of $\Theta$ we have that:
\begin{gather}
\mathbb{E}\Big[\Big\|\frac{1}{N}\sum_{j\neq i}^N\Theta(x_k^i,x^j_k,\zeta)-\frac{1}{N}\sum_{j\neq i}^N\Theta(\bar{x}_k^i,x^j_k,\zeta)\Big\|\Big]\nonumber\\+\mathbb{E}\Big[\Big\|\frac{1}{N}\sum_{j\neq i}^N\Theta(\bar{x}_k^i,x^j_k,\zeta)-\int_{\mathcal{S}}\Theta(\bar{x}_k^i,\bar{x}^j_k,\zeta)\Big\|\Big]\nonumber+\mathbb{E}\Big[\Big\|\frac{1}{N}\sum_{j\neq i}^N\Theta(\bar{x}_k^i,\bar{x}^j_k,\zeta)-\int_{\mathcal{S}}\Theta(\tilde{x}_k,x',\zeta)m(dx')\Big\|\Big]\nonumber
\\
\leq c_1\|x^i_k-\bar{x}^i_k\|+c_2\frac{1}{N}\sum_{j\neq i}^N\|x^j_k-\bar{x}_k^j\|+\mathbb{E}\Big[\Big\|\frac{1}{N}\sum_{j\neq i}^N\Theta(\bar{x}_k^i,\bar{x}^j_k,\zeta)-\int_{\mathcal{S}}\Theta(\tilde{x}_k,x',\zeta)m(dx')\Big\|\Big],\nonumber
\end{gather}
where $c_1,c_2>0$ are arbitrary (Lipschitz) constants. After  summing over $i$ and dividing by $N$ and, using (\ref{indhyp2})  we deduce the required result.\newline
\indent For part (ii) we exploit the Lipschitzianity of the function $L$. Moreover, since the policy $\pi$ can be expressed as $\pi=\Gamma(x^i,x^{-i})$ we express the instantaneous function $L$ as $\hat{L}:\mathcal{S}^i\times\mathcal{S}^{-i}$ using \cite{Xu2015}.\newline
\indent We now observe the following estimate:
\begin{gather*}
\mathbb{E}\Big[\Big\|\frac{1}{N}\sum_{j\neq i}^N\hat{L}(\bar{x}_k^i,x^{i}_k)-\int_{\mathcal{S}}\hat{L}(\tilde{x}_k,x')m(dx')\Big\|\Big]
\nonumber\\\leq
c_1\|x^i_k-\bar{x}^i_k\|+c_2\frac{1}{N}\sum_{j\neq i}^N\|x^j_k-\bar{x}_k^j\|+\Bigg[\sum_{l=1}^d{\rm var}\Bigg(\frac{1}{n}\sum_{j\neq i}^N\Theta(x^i_k,x^{-i}_k,\cdot)\Bigg)\Bigg]^{\frac{1}{2}}.
\end{gather*}
where $c_1,c_2>0$ are arbitrary constants. \newline \indent As in part (i), after summing in $i$ and dividing by $N$ and, using the result in part (i) we deduce the result.
\end{proof}

\subsubsection*{Proof of Theorem 2}
\begin{proof}
The theorem is proved by demonstrating that given some initial belief $\tilde{m}_0$ of the distribution $m$, each agent can generate an iterative sequence $\{\tilde{m}^n\}_{\{n\in\mathbb{N}\}}$ s.th. $\{\tilde{m}^n\}_{\{n\in\mathbb{N}\}}\to m$ as $n\to\infty$. Thereafter, we deploy a \emph{two-timescales method} \cite{borkar,Leslie2006} so that the the beliefs of the distribution (at a given state) are updated slowly and thus are quasi-static from the perspective of the updating procedure for the policy. In particular, in the following, we show that provided the updating procedure to the agent's belief of the distribution $\tilde{m}$ are performed at a sufficiently pace relative to the updates to the policy, the problem is reduced to an optimisation problem for each agent where at each stage of the iteration each agent plays an approximate best-response given some belief of the distribution of the agents (evaluated at the agent's own location). The agents improve their policies using a stochastic gradient procedure which, with an update to the distribution that appears static, produces enough traction to enable convergence to the optimal policy.     

Consider first, each agent's optimal value function which, given a belief of the distribution $\tilde{m}$ is given by the following:
\begin{equation}
v^{\pi,\tilde{m}}(x)\triangleq\sup_{\pi'\in \Pi} J[x,\pi',\tilde{m}_{x}],    \label{payofffunctionmfg belief}
\end{equation}
Suppose also that each agent has a belief $\tilde{m}^n$ of the distribution $\tilde{m}$ over $n=0,1,2,\ldots$, iterations which are updated using a procedure that we will later specify (this can in fact be viewed as a parameter of the value function at iteration $n$), then following (\ref{payofffunctionmfg belief}) we can define the $n^{th}$ iterate of $v$ by the following expression:
\begin{equation}
v^{\pi,\tilde{m}}_n(x)\triangleq\sup_{\pi'\in \Pi} J[x,\pi',\tilde{m}^n_{x}],   \qquad \forall x\in\mathcal{S} \label{payofffunctionmfg kth belief}
\end{equation}
Using the boundedness of $L$ we observe that $\underset{n}{\sup}|v_n^{\pi,\tilde{m}}(x)|=\underset{n,\pi'\in \Pi}{\sup} J[x,\pi',\tilde{m}^n_{x}]\leq \sum_{k=t}^T \|L\|_{\infty}<\infty$, so that the sequence $\{v^{\pi,\tilde{m}}_n(x)\}_{n\geq 1}$ consists of bounded terms.

Let us analogously define the $n^{th}$ \emph{fictitious best-response} as:
\begin{equation}
\tilde{\pi}^n\in\arg\hspace{-0.45 mm}\sup_{\hspace{-2.3 mm}\pi'\in \Pi} v^{\pi',\tilde{m}^n}_n(x),\qquad \forall x\in\mathcal{S}    \label{kth response}
\end{equation}
Suppose that each iteration each agent updates their belief over the distribution via the procedure:
\begin{align}
\hspace{9 mm} \tilde{m}^{n+1}_x=\tilde{m}^{n}_x+ \beta(n)[B(a^n,\tilde{m}^n_x)+M_{n+1}^{(\tilde{m})}],\qquad n=0,1,\ldots, \qquad \forall x\in\mathcal{S},\quad a^n\sim v_n \label{field update}
\end{align}
where $\beta(n)$ is a positive step-size function s.th. $\sum_n\beta(n)=\infty$ and $\sum_n\beta(n)^2<\infty$, $B:\mathcal{A}\times\mathbb{R}\to \mathbb{R}$ is some well-behaved function (recall that $m_x\in\mathbb{R}^d$ is the agent density evaluated at the point $\in\mathcal{S}\subset \mathbb{R}^d$) and $M^{(\tilde{m})}_{n}$ is a martingale difference sequence\footnote{W.r.t. the increasing family of $\sigma-$ algebras $\sigma(x_l,M^{(\tilde{m})}_l,l\leq n)$.}. In particular, we assume that $B$ is Lipschitz continuous in both variables and satisfies a growth condition i.e $\| B(x,y)\|\leq c(1+\| (x,y)\|)$ for any $(x,y)\in \mathcal{A}\times\mathbb{P}(\mathcal{H})$ and $\| B(x,y)-B(w,z)\|\leq d\|(x,y)-(w,z)\|$ for any $(w,z),(x,y)\in \mathcal{A}\times\mathbb{P}(\mathcal{H})$ for some constants $c,d>0$.

We note that since the update can be viewed as an updating procedure over a parameterisation of the value function, the update in (\ref{field update}) can be viewed as a critic update of the value function $v^{\tilde{m}}_n$ (parameterised by $\tilde{m}_1,\tilde{m}_2,\ldots$) within the context of a mean field game. It can be straightforwardly shown that $\sum_n\beta(n)M_{n}^{(\tilde{m})}<\infty$ (see for example, Appendix C in \cite{borkar}).

Since the action $a^n$ is sampled from the policy $\pi^n$ and, using (\ref{kth response}), we can rewrite (\ref{field update}) as the following expression:
\begin{equation}
\qquad\qquad\qquad \tilde{m}^{n+1}_x=\tilde{m}^{n}_x+ \beta(n)[B'(v_n,\tilde{m}^n)+M_{n+1}^{(\tilde{m})}],\qquad n=0,1,\ldots \quad \qquad \forall x\in\mathcal{S}   \label{field update 2}
\end{equation}
for some map $B'$ which maps to and from the same spaces and satisfies the same assumptions as $B$ (we have suppressed the super-indices on the function $v$).

Our goal is to show that the following expression holds $\forall x\in\mathcal{S}$:
\begin{equation}
\lim_{n\to\infty}v^{\tilde{\pi}^n,\tilde{m}^n}_n(x)=v^{\lambda(m_x),m_x}\qquad a.s.
\end{equation}
where $(\lambda(m),m)$ is a fixed point solution to the MFG system (B) (i.e. $\lambda(m)\in\sup_{\pi'\in\Pi}v^{\pi'(m_x),m_x}(x)$).

Let us now consider the following update procedure for $v_n$:
\begin{equation}
v_{n+1}(x)=v_{n}(x)+ \alpha(n)[A(\pi^n,\tilde{m}^n_x)+{M}_{n+1}^{(\pi)}],\qquad n=0,1,\ldots\quad \forall x\in\mathcal{S}    \label{policy update slow}
\end{equation}
where $M^{(\pi)}_{n}$ is a martingale difference sequence and $A:\Pi\times\mathbb{R}\to \mathbb{R}$ is a Lipschitz continuous function  that satisfies a growth condition in both variables and where $\alpha$ is a positive step-size function s.th. $\sum_n\alpha(n)=\infty$ and $\sum_n\alpha(n)^2<\infty$ and is chosen such that  $\alpha(n)\backslash \beta(n) \sim \mathcal{O}(\frac{1}{n^p})$ for some $p>1$ so that the belief $\tilde{m}^n$ is updated slowly relative to $\pi$ and $\alpha(n)\backslash \beta(n)\to 0$ as $n\to \infty$. In a similar way it can be shown that $\sum_n\alpha(n)M_{n}^{(\pi)}<\infty$. 

We can now rewrite the update process over $\tilde{m}$ as the following:
\begin{equation}
\tilde{m}^{n+1}_x=\tilde{m}^{n}_x+ \alpha(n)[\hat{B}'(v_n,\tilde{m}^n_x)+\hat{M}_{n+1}^{(\tilde{m})}],\qquad n=0,1,\ldots    \label{field update slow}
\end{equation}
where $\hat{B}'(v_n,\tilde{m}^n_x)\triangleq\frac{\beta (n)}{\alpha(n)}B'(v_n,\tilde{m}^n_x)$ and $\hat{M}_{n+1}^{(\tilde{m})}\triangleq\frac{\beta (n)}{\alpha(n)}M_{n+1}^{(\tilde{m})}$. We note that the sequence $\{\tilde{m}^{n}_x\}_{n\geq 1}$ is quasi-static w.r.t. the sequence $\{v_n\}_{n\geq 1}$, moreover the sequence $\{v_n\}_{n\geq 1}$ converges when $\tilde{m}^n_x$ is fixed at a particular $n\in\mathbb{N}$ (the result follows since it can be shown that $v_n$ strictly increases whenever $v_n$ is suboptimal).

Let us firstly recall that $\underset{n}{\sup}|v_n^{\pi,\tilde{m}}(x)|< \infty$ and $\sum_n\alpha(n)M^{(\pi)}_{n}<\infty$. Analogously, we have that $\sum_n\alpha(n)\hat{M}_{n}^{(\tilde{m})}=\sum_n\beta(n)M^{(\tilde{m})}_{n}<\infty$ moreover, since $\tilde{m}^n_x$ is defined over a bounded domain and, using Assumption 2 we deduce that $\underset{n}{\sup}|\tilde{m}^n_x|< \infty$. We can therefore apply theorem 2, ch. 6 in \cite{borkar} from which find that $\lim_{n\to\infty}(v_n^{\tilde{\pi}^n,\tilde{m}},\tilde{m}^n)\to (v^{\lambda(m),m},m)$ a.s. after which, using (\ref{kth response}), we deduce the thesis.  

\end{proof}

%

\end{document}